\documentclass[osid.sty]{article} 

\usepackage[utf8]{inputenc} 

\usepackage{geometry} 
\usepackage{graphicx} 
\usepackage{amsmath}
\usepackage{amssymb}
\usepackage{booktabs} 
\usepackage{array} 
\usepackage{paralist} 
\usepackage{verbatim} 
\usepackage{subfig} 
\usepackage{float}
\usepackage{fancyhdr} 
\usepackage{sectsty}
\allsectionsfont{\sffamily\mdseries\upshape} 
\usepackage[nottoc,notlof,notlot]{tocbibind} 
\usepackage[titles,subfigure]{tocloft} 

\newtheorem{lemma}{Lemma}

\newtheorem{proposition}{Proposition}

\newtheorem{remark}{Remark}

\newcommand{\qed}{\hfill $\Box$\medskip}

\newcommand{\CI}{{\mathbb{C}}}

\newcommand{\RI}{{\mathbb{R}}}

\newcommand{\HI}{{\mathbb{H}}}

\DeclareMathOperator{\e}{e}

\newcommand{\aop}{\hat{a}}
\newcommand{\cop}{\hat{a}^{\dag}}
\newcommand{\gao}[1]{\hat{#1}}
\newcommand{\gco}[1]{\hat{#1}^{\dag}}

\newcommand{\elr}[2]{\left\vert #1 \right\rangle \left\langle #2 \right\vert}
\newcommand{\bk}[2]{\left< #1\left\vert\right. #2\right>}
\newcommand{\kb}[2]{\left\vert#1\right\rangle\left\langle#2\right\vert}
\newcommand{\bra}[1]{\left< #1 \right\vert}
\newcommand{\ket}[1]{\left\vert #1 \right>}

\newcommand{\modu}[1]{\left| #1 \right|}
\newcommand{\bkmv}[3]{\left< #1 \right\vert#2 \left\vert #3 \right>}

\newtheorem{mydef}{Definition}[section]
\newtheorem{coro}{Corollary}[section]

\title{Entanglement of Identical Particles}

\author{F. Benatti\\{\footnotesize\it Department of Physics, University of Trieste, Italy \& INFN, Sezione di Trieste, email: benatti@ts.infn.it}\\[2ex]
R. Floreanini\\{\footnotesize\it INFN, Sezione di Trieste, email: floreanini@ts.infn.it}\\[2ex]
K. Titimbo\\{\footnotesize\it Department of Physics, University of Trieste, Italy \& INFN, Sezione di Trieste, email: titimbo@ts.infn.it}}

\begin{document}

\maketitle

\begin{abstract}
Unlike for bipartite states consisting of distinguishable particles, in the case of identical parties the notion of entanglement is still under debate.
In the following, we review two different approaches to the entanglement of systems consisting of two Bosons or Fermions; the first approach is based on the particle aspect typical of first quantization and identifies separable pure states as those that allow to
assign two orthogonal single particle vector states to both parties.
The second approach makes full use of the mode aspect of second quantization whereby separability can be formulated as absence of non-local correlation among two different sets of modes. While the first approach applies to pure states only, the second one is more general and characterizes generic entangled states. In the following, we shall show that the mode-based approach indeed contains the particle-based one.
\end{abstract}

\section{Introduction}

Quantum entanglement is perhaps the most intriguing consequence of the linear structure of quantum mechanics; it refers to statistical correlations among sub-systems of compound quantum systems that forbid the attribution of properties to the individual parties even if these are far apart from each other.
Since the EPR gedankenexperiment \cite{einapr1935}, the scientific perception of the phenomenon changed from being an argument against the compatibility of quantum mechanics with special relativity to becoming a physical resource usable to perform tasks that would be impossible in a completely classical world.
The properties of entangled states have now applications in as many different fields as quantum teleportation, quantum cryptography, and quantum computation.
Nevertheless, while quantum entanglement has a clear formulation in relation to distinguishable-particle systems \cite{horrhh2009}, the entanglement of identical particles is still lacking an agreed upon status \cite{Zanardi1}--\cite{bala}.

In this introduction we shall present the issue at stake by means of the simplest possible setting, that of a system consisting of $2$ two-level systems (two qubits) each of which described by the Hilbert space $\mathbb{C}^2$. These qubits will firstly be treated as distinguishable and then as identical.

Let $\ket{\leftrightarrow}$ and $\ket{\updownarrow}$ denote two orthogonal one qubit vector states. In the standard approach, the separability or entanglement of a two qubit vector state is judged with respect to the underlying tensor product structure of the Hilbert space $\mathbb{C}^2\otimes\mathbb{C}^2$. This tensor product structure embodies the a-priori knowledge of which is the first qubit and which the second one, together with their corresponding individual observables. Among them, beside describing a so-called pure state, any one-dimensional projection $P_\varphi=\kb{\varphi}{\varphi}$ also corresponds to the property that the system may be found in the state $\ket{\varphi}\in \mathbb{C}^2$. This property is possessed by the system if its state $\rho$ is such that
\begin{equation}
\label{1q0}
{\rm Tr}(\rho P_\varphi)=1 \,
\end{equation}
that is if and only if the density matrix $\rho$ describing the system state is exactly $P_\varphi$.

Individual observables are local with respect to the Hilbert space tensor product structure; indeed, single qubit operators are $2\times 2$ matrices $M_2(\mathbb{C})$ and these are embedded in the two qubit algebra of $4\times 4$ matrix algebra $M_4(\CI)=M_2(\mathbb{C})\otimes M_2(\mathbb{C})$
as elements of the sub-algebras $M_2(\mathbb{C})\otimes\mathbb{I}\subset M_4(\CI)$ and $\mathbb{I}\otimes M_2(\mathbb{C})\subset M_4(\CI)$.

With respect to such natural tensor product structures, separable vector states are identified as the tensor products of single particle vector states; simple instances being
\begin{equation}
\label{0}
\ket{\leftrightarrow}\otimes\ket{\updownarrow}\ ,\quad\ket{\updownarrow}\otimes\ket{\updownarrow}\ ,\quad
\frac{\ket{\leftrightarrow}\otimes\ket{\updownarrow}+\ket{\updownarrow}\otimes\ket{\updownarrow}}{\sqrt{2}}=
\frac{\ket{\leftrightarrow}+\ket{\updownarrow}}{\sqrt{2}}\otimes\ket{\updownarrow}\ ,
\end{equation}
while the following ones are simple local operators addressing each one of the particles, independently,
$$
P_1=\ket{\leftrightarrow}\bra{\leftrightarrow}\otimes \mathbb{I}\ ,\quad P_2=\mathbb{I}\otimes\ket{\updownarrow}\bra{\updownarrow}\ ,\quad
P_1\otimes P_2\ ,
$$
where $\mathbb{I}$ denotes the identity matrix.

If the compound system is in the state $\ket{\leftrightarrow}\otimes\ket{\updownarrow}$, the orthogonal projection $P_1$, $P_2$ and $P_1\otimes P_2$ correspond to possessed properties: $P_1$ to the first qubit being in the state $\ket{\leftrightarrow}$, $P_2$ to  the second qubit being in the state $\ket{\updownarrow}$.  Indeed, $\ket{\leftrightarrow}\otimes\ket{\updownarrow}$ is eigenstate of $P_{1,2}$ and $P_1\otimes P_2$ so that
$$
\bra{\leftrightarrow}\otimes\bra{\updownarrow}\Big(P_1\otimes\mathbb{I}\Big)\ket{\leftrightarrow}\otimes\ket{\updownarrow}=
\bra{\leftrightarrow}\otimes\bra{\updownarrow}\Big(\mathbb{I}\otimes P_2\Big)\ket{\leftrightarrow}\otimes\ket{\updownarrow}=
\bra{\leftrightarrow}\otimes\bra{\updownarrow}\Big(P_1\otimes P_2\Big)\ket{\leftrightarrow}\otimes\ket{\updownarrow}=1\ .
$$
Instead, any sum of local operators, like the so-called $CNOT$ operation,
$$
U_{CNOT}=\ket{\leftrightarrow}\bra{\leftrightarrow}\otimes \mathbb{I}+\ket{\updownarrow}\bra{\updownarrow}\otimes
\Big(\ket{\leftrightarrow}\bra{\updownarrow}+\ket{\updownarrow}\bra{\leftrightarrow}\Big) \ ,
$$
cannot be reduced to a single tensor product.
Acting on the third separable state in (\ref{0}), $U_{CNOT }$ creates the (maximally) entangled state
\begin{equation}
\label{1}
\ket{\Psi}=U_{CNOT}\frac{\ket{\leftrightarrow}+\ket{\updownarrow}}{\sqrt{2}}\otimes\ket{\updownarrow}
=\frac{1}{\sqrt{2}}\Big(\ket{\leftrightarrow}\otimes\ket{\updownarrow} + \ket{\updownarrow}\otimes\ket{\leftrightarrow}\Big)\ .
\end{equation}Unlike separable vector states, $\ket{\Psi}$ is such that none of the properties associated with  $P_1\otimes\mathbb{I}$, $\mathbb{I}\otimes P_2$ and $P_1\otimes P_2$ can be attributed to the system in such a state; in fact,
$$
\bra{\Psi}P_1\otimes\mathbb{I}\ket{\Psi}=\bra{\Psi}\mathbb{I}\otimes P_2\ket{\Psi}=\bra{\Psi}P_1\otimes P_2\ket{\Psi}=\frac{1}{2}\ .
$$
In full generality, the projections $P_1\otimes\mathbb{I}$, $\mathbb{I}\otimes P_2$ and  $P_1\otimes P_2$, where $P_{1,2}$ project onto not necessarily orthogonal single particle vector states $\ket{\varphi_1}$, $\ket{\varphi_2}$, describe  possessed properties by a two qubit system in a vector state $\ket{\Psi}$ if and only if $\ket{\Psi}=\ket{\varphi_1}\otimes\ket{\varphi_2}$, that is if the vector state is separable. This follows at once from the Cauchy-Schwartz inequality that  must be saturated by $\rho=\kb{\Psi}{\Psi}$
in order to satisfy \eqref{1q0}.
To summarize, the entangled vector states of two distinguishable particles are such that no individual properties can be attributed to each one or both of the constituent parties.

Suppose now the two qubits to be identical. The fact that the two parties cannot be distinguished implies that their vector states must be symmetric in the single particle states for Bosons and anti-symmetric for Fermions. With reference to the single particle orthonormal basis $\{\ket{\leftrightarrow},\ket{\updownarrow}\}$ the Bosonic sector of the Hilbert space $\mathbb{C}^2\otimes\mathbb{C}^2$ is three dimensional and linearly spanned by the symmetric orthogonal vectors
\begin{equation}
\label{states}
\ket{\leftrightarrow}\otimes\ket{\leftrightarrow}\ ,\quad
\ket{\updownarrow}\otimes\ket{\updownarrow}\ ,\quad
\ket{\Psi}=\frac{1}{\sqrt{2}}\Big(\ket{\leftrightarrow}\otimes\ket{\updownarrow} + \ket{\updownarrow}\otimes\ket{\leftrightarrow}\Big)\ ,
\end{equation}
while the Fermionic sector is one dimensional and given by the anti-symmetric vector
$$
\ket{\Phi}=\frac{1}{\sqrt{2}}\Big(\ket{\leftrightarrow}\otimes\ket{\updownarrow} - \ket{\updownarrow}\otimes\ket{\leftrightarrow}\Big)\ .
$$
Note that particle identity does not exclude that individual properties might be attributed to the two parties; the only constraint is that no property can be attributed to a definite party otherwise it could be used to distinguish it from the other.
Nevertheless, single-particle properties are still described by one-dimensional projections; therefore, in order to implement the latter constraint, the property corresponding to one qubit being in the generic state $\ket{\varphi}$, must be represented by the symmetric (not one-dimensional) projection
$$
\mathcal{E}_{P} = P\otimes \Big(\mathbb{I}-P\Big) + \Big(\mathbb{I}-P\Big)\otimes P + P \otimes P\  ,\quad P=\ket{\varphi}\bra{\varphi}\ .
$$
Similarly, the projector corresponding to the two qubits possessing two different properties must be
$$
P^{symm}_{12}=P_1\otimes P_2 + P_2\otimes P_1\ ,\qquad P_{1}=\kb{\varphi_{1}}{\varphi_{1}}\ ,\quad 
P_{2}=\kb{\varphi_{2}}{\varphi_{2}}\ .
$$
This is a projection if and only if $P_1P_2=0$, whence $P^{symm}_{12}=\mathcal{E}_{P_1}\mathcal{E}_{P_2}$.
Therefore, despite the formally entangled structure of a state as $\ket{\Psi}$, properties can be attributed to both its parties. Indeed,
\begin{equation}
\label{symprop2}
\bra{\Psi}\mathcal{E}_{P_1}\ket{\Psi}=\bra{\Psi}\mathcal{E}_{P_2}\ket{\Psi}=\bra{\Psi}P^{sym}_{12}\ket{\Psi}=1 \ ;
\end{equation}
however, it is not known which party has which property.
In \cite{ghigmw2002,GhiMar}, the separability of  Bosonic and Fermionic $2$-particle vector states is identified with the possibility of attributing properties to both parties. It was then proved that separable vector states of two identical particles are either symmetric or anti-symmetric combinations of tensor products of orthogonal single particle vector states.

The above approach is pursued within a first quantization context, while in \cite{benffm2010,benffm2012} a second quantized point of view is taken
whereby separability and entanglement are defined with reference to the algebraic structure of Bose and Fermi systems rather than in relation to the attribution of individual properties. As already stressed, being linearly spanned by either symmetric (Bosons) or anti-symmetric (Fermions) tensor products of single particle vector states, Bosonic or Fermionic Hilbert spaces are no longer embodied with an {\it a priori} tensor product structure.
However, one observes that, especially in many body quantum systems, the primary object is the algebra of operators rather than its representations on particular Hilbert spaces. 

In the case of identical particles, the algebra of operators $\mathcal{A}$ is constructed by means of polynomials in annihilation and creation operators $\hat{a}_i$, $\hat{a}_i^\dag$, $i$ running over a set $I$ numbering the elements of an orthonormal basis in the single particle Hilbert space $\HI$, or an associated choice of modes~\footnote{\label{footnote1}
Actually, since the Bose creation and annihilation operators are unbounded, one has to resort to the so-called Weyl algebra generated by exponential $\exp(\alpha\,\hat{a}(f)+\alpha^*\,\hat{a}^\dag(f))$ whence polynomials are generated by differentiating with respect to the complex parameter $\alpha$ \cite{BraRob}.}. They satisfy the canonical commutation relations (CCR)
\begin{equation}
\label{CCR}
[\hat{a}_i\,,\,\hat{a}_j^\dag]=\hat{a}_i\,\,\hat{a}_j^\dag\,-\,\hat{a}^\dag_j\,\hat{a}_i=\delta_{ij}\ ,
\end{equation}
in the Bosonic case, and, in the Fermionic case, the canonical anti-commutation relations (CAR)
\begin{equation}
\label{CAR}
\{\hat{a}_i\,,\,\hat{a}_j\}=\hat{a}_i\,\hat{a}_j^\dag\,+\,\hat{a}_j^\dag\,\hat{a}_i=\delta_{ij}\ .
\end{equation}

As already seen, in the case of two distinguishable qubits, the algebraic structure 
$M_2(\mathbb{C})\otimes M_2(\mathbb{C})$ selects $M_2(\mathbb{C})\otimes\mathbb{I}$ and $\mathbb{I}\otimes M_2(\mathbb{C})$ as natural sub-algebras thereby identifying local operators with tensor products $A_1\otimes A_2$. The salient feature here is the fact that $A_1\otimes \mathbb{I}$ and $\mathbb{I}\otimes A_2$ commute: this corresponds to their (algebraic) independence since commutativity excludes mutual influences when two of these observables are simultaneously measured.

In the second quantization formalism, pairs of mutually commuting sub-algebras can easily be constructed by partitioning the index set $I$ into two disjoint subsets $I_{1,2}$.
From \eqref{CCR}, the two sub-algebras $\mathcal{A}_{1,2}$ generated by  $\{a_i,a^\dag_j\}_{(i,j)\in I_1}$, respectively $\{a_i,a^\dag_j\}_{(i,j)\in I_2}$ commute for Bosons. For Fermions, from \eqref{CAR}, one sees that they do so only if at least one of the them is generated by polynomials of even degree, that is
only if at least one of the two sub-algebras $\mathcal{A}_{1,2}$ is \textit{even}.
Indeed, if $\mathcal{A}_1$ is an even algebra, the fact that $[A_1\,,\,A_2]=0$ for all $A_{1,2}\in\mathcal{A}_{1,2}$ follows from the relation
\begin{equation}
\label{CAR-CCR}
[AB\,,\,C]\,=\,A\,\{B\,,\,C\}-\{A\,,\,C\}\,B\ ,
\end{equation}
while nothing can be said in general of commutators like $[\hat{a}_i\,,\,\hat{a}_j^\dag]$ by knowing that $\{\hat{a}_i\,,\,\hat{a}_j^\dag\}=0$.
With this proviso, one can extend the notion of locality to a system of identical particles through the following definitions:%
\footnote{In the case of Fermions, a more general definition of bipartition is possible, not requiring that the two subalgebras $\mathcal{A}_{1,2}$ commute. Nevertheless, the notion of pure state separability that follows
is still characterized by the condition (\ref{3}) below. For more details, see \cite{benffm2013}.}
\begin{enumerate}
\item
a bipartition $(\mathcal{A}_1,\mathcal{A}_2)$ of the algebra of operators $\mathcal{A}$ is any pair of commuting sub-algebras $\mathcal{A}_1$, $\mathcal{A}_2$;
\item
an operator $A\in \mathcal{A}$ is said to be local with respect to the bipartition 
$(\mathcal{A}_1,\mathcal{A}_2)$, if it is of the form
$A=A_1A_2$, where $A_1\in\mathcal{A}_1$ and $A_2\in\mathcal{A}_2$.
\end{enumerate}
Then, one defines separable with respect to a bipartition $(\mathcal{A}_1,\mathcal{A}_2)$ those $\ket{\Psi}$ such that
\begin{equation}
\label{3}
\bra{\Psi}A_1A_2\ket{\Psi}=\bra{\Psi}A_1\ket{\Psi}\,\bra{\Psi}A_2\ket{\Psi}\ ,
\end{equation}
for all possible $A_{1,2}$ belonging to the commuting sub-algebras $\mathcal{A}_{1,2}$, in the case of a Boson system, while $A_{1,2}$ must belong to the even components $\mathcal{A}^{ev}_{1,2}$ in the case of
a Fermi system (see Section \ref{sec2.2} for further details on this point).

Evidently, in the second quantized approach the notion of local operators and separable states, hence of non-local operators and entangled states, depend on the reference pair of commuting sub-algebras.
As we will see in the following, such a definition of separability is by no means restricted to vector states, but extends to mixed states.
Furthermore, in the case of two distinguishable qubits, separable bipartite vector states  are those assigning factorized mean
values to local operators $A_1\otimes A_2$; {\rm i.e.} \eqref{3} reduces to the known definition of separable vector states for distinguishable qubits.

As an illustration of the above approach, consider two identical qubits with Bosonic character; in the second quantization formalism, the single particle orthonormal
states $\ket{\leftrightarrow}$ and $\ket{\updownarrow}$ correspond to two possible Bosonic modes and are generated by acting on the vacuum state $\ket{0}$ with creation operators,
$\gco{a}_1$ and $\gco{a}_2$, i.e.,
$$
\gco{a}_1\ket{0} = \ket{\leftrightarrow}\ ,\quad \gco{a}_2\ket{0} = \ket{\updownarrow}\ ,
$$
the Bosonic character being expressed by
$$
[\gao{a}_1\,,\,\gco{a}_1]=[\gao{a}_2\,,\,\gco{a}_2]=1\ ,\quad [\gao{a}_1\,,\,\gco{a}_2]=[\gao{a}_2\,,\,\gco{a}_1]=0\ .
$$
This is the simplest possible non-trivial Boson system: two Bosons, each of them having two orthonormal states (or modes) at disposal;
notice that, due to the Bosonic character of the degrees of freedom, the Hilbert space is $3$ dimensional and linearly spanned by the
vectors
$$
\ket{\leftrightarrow,\leftrightarrow}=\frac{(\gco{a}_1)^2}{\sqrt{2}}\ket{0}\ ,\quad\,
\ket{\updownarrow,\updownarrow}=\frac{(\gco{a}_2)^2}{\sqrt{2}}\ket{0}\, \quad\ket{\Psi}=\gco{a}_1 \gco{a}_2 \ket{0}\ ,
$$
which are the second quantized version of (\ref{states}).

All operators in the algebra $\mathcal{A}_1$ generated by the polynomials in $\gao{a}_1$ and $\gco{a}_2$ and those in the algebra $\mathcal{A}_2$ generated by polynomials in $\gao{a}_2$ and $\gco{a}_2$ commute. It thus follows that the action of $\gco{a}_1\gco{a}_2$ on the vacuum is local with respect to the two sub-algebras $\mathcal{A}_1$ and $\mathcal{A}_2$.

According to the physical intuition outlined above, the state resulting from a local action on the vacuum (which is separable since there is nothing in the vacuum that can be entangled) should result separable from the point of view of the sub-algebras $\mathcal{A}_{1,2}$.
The separability of $\ket{\Psi}=\hat{a}_1^\dag\hat{a}^\dag_2\ket{0}$ with respect to the bipartition $(\mathcal{A}_1,\mathcal{A}_2)$ is confirmed by the absence of correlations between $\mathcal{A}_{1,2}$ carried by such a state: (\ref{3}) is indeed satisfied. This can be seen as follows: polynomials in creation and annihilation $P(\gao{a}_i,\gco{a}_i)$ can always be written as sums of monomials of the form $(\gao{a}_i)^p(\gco{a}_i)^q$. On the other hand, from the CCR,
\begin{eqnarray*}
\bra{\Psi}\gao{a}_1^p(\gco{a}_1)^q\,\gao{a}_2^m(\gco{a}_2)^n\ket{\Psi}&=&\bra{0}\gao{a}_1^{p+1}(\gco{a}_2)^{q+1}\,\gao{a}_2^{m+1}(\gco{a}_2)^{n+1}\ket{0}\\
&=&
\bra{0}\gao{a}_1^{p+1}(\gco{a}_1)^{q+1}\ket{0}\,\bra{0}\gao{a}_2^{m+1}(\gco{a}_2)^{n+1}\ket{0}\\
&=&
\bra{\Psi}\gao{a}_1^{p}(\gco{a}_1)^{q}\ket{\Psi}\,\bra{\Psi}\gao{a}_2^{m}(\gco{a}_2)^{n}\ket{\Psi}\ .
\end{eqnarray*}
Therefore, the state $\ket{\Psi}$ in (\ref{1}), which is maximally entangled for distinguishable qubits, is instead separable for identical (Bosonic) qubits also in the second quantization approach.

While the second quantized approach applies to generic states of identical particle systems, the first quantized approach, based on the particle aspect of first quantization, does not cover the case of mixed states.  The goal of the present work is to provide a detailed comparison of the first and second quantized approaches in the case of vector states of identical particle systems.
We will thus give first a comprehensive introduction to the two methods that have so far been sketched, extending the second quantization approach to cover the Fermionic case; finally, we show that the first quantized particle-based approach is contained in the second quantized mode-based approach.

\section{Pure State Bipartite Entanglement for Identical Particles}
\label{sec2}

As outlined in the introduction, in composite quantum systems the standard notion of locality that identifies separable and entangled states relies on the tensor product structure of the Hilbert space. In the case of indistinguishable particles such structure is no longer available due to the required symmetrization or anti-symmetrization of vector states.

\subsection{First Quantization}
\label{sec2.1}

The question addressed in \cite{ghigmw2002,GhiMar} is whether two-particle Bosonic and Fermionic vector states are automatically
entangled because of their non-tensor product structure. The answer is: never when these arise from the symmetrization (Bosons) and anti-symmetrization (Fermions) of tensor products of orthogonal single particle vector states.
The argument on which the answer is based follows from the possibility of identifying properties objectively possessed by the individual parties.

For distinguishable particles in factorized,  separable vector states $\ket{\Phi}=\ket{\varphi_1}\otimes\ket{\varphi_2}$, a well-defined state vector $\ket{\varphi_{1,2}}$ is assigned to each component sub-system; since such states are simultaneous eigenstates of a
complete set of commuting observables, it is possible to predict with certainty the measurement outcomes of this set of operators. These outcomes constitute a complete set of properties that can be legitimately thought as possessed by each particle. Indeed, as already discussed in the introduction, the single particle projections $P_{1,2}=\kb{\varphi_{1,2}}{\varphi_{1,2}}$ are such that the projections $P_1\otimes\mathbb{I}$, $\mathbb{I}\otimes P_2$ and $P_1\otimes P_2$ all have mean value $1$ with respect to $\ket{\Phi}$. Namely these properties are attained with probability $1$.
On the other hand, when the system is in an entangled vector state, definite state vectors cannot be associated with any constituent and, therefore, one cannot claim that the parties objectively possess a complete set of properties whatsoever.

In order to extend these arguments to a composite system $\mathcal{S}$ consisting of two identical particles described by the pure normalized state $\ket{\Psi}$, one must note that a complete set of properties cannot be attributed to a definite party, otherwise these properties would distinguish it from the other, identical constituent. The only  meaningful assertion is that one of the two parties has a complete set of properties;
if this latter corresponds to a single particle one dimensional projection operator $P$, the assertion is mathematically translated into
\begin{equation}
\label{symmprop2}
\bkmv{\Psi}{\mathcal{E}_{P}}{\Psi} =1\ ,\quad
\mathcal{E}_{P}=P\otimes (\mathbb{I}-P) + (\mathbb{I}-P)\otimes P +P\otimes P = P\otimes \mathbb{I} + \mathbb{I}\otimes P - P \otimes P\ .
\end{equation}
Consider the first expression of the operator in \eqref{symmprop2}: it  is a symmetric projection where the first term refers to the first particle, and not the second, having the complete set of properties associated with $P$,
the second term has the roles of the two particles exchanged and the third term refers to both particles having the considered complete set of properties.
Since the three terms are mutually orthogonal, the condition $\bkmv{\Psi}{\mathcal{E}_P}{\Psi}=1$ implies that at least one particle has the complete set of  properties described by the single particle projection $P$.

\begin{mydef}
\label{gs}
The identical constituents $\mathcal{S}_{1}$ and $\mathcal{S}_{2}$ of a composite quantum system $\mathcal{S}=\mathcal{S}_{1}+\mathcal{S}_{2}$ are separable when both constituents possess a complete set of properties.
\end{mydef}

The fact that one constituent possesses a complete set of properties has immediate consequences on the form of the state $\ket{\Psi}$.

\begin{proposition}
One of the identical constituents of a composite quantum system $\mathcal{S}$, described by the pure normalized state $\ket{\Psi}$ has a complete set of properties associated with a one-dimensional projector $P=\elr{\phi_0}{\phi_0}$ onto a single-particle vector state if and only if $\ket{\Psi}$ is obtained by symmetrizing (Bosons) or anti-symmetrizing (Fermions) a tensor product vector state.
\end{proposition}

\begin{proof}
A general two-particle vector state can be written as
$$
\ket{\Psi}=c_{00} \ket{\phi_{0}}\otimes \ket{\phi_0} + \sum_{j>0} c_{0j} \ket{\phi_0}\otimes\ket{\phi_{j}} + \sum_{i>0} c_{i0}\ket{\phi_{i}}\otimes\ket{\phi_0} + \sum_{i,j \neq 0} c_{ij}\ket{\phi_{i}}\otimes\ket{\phi_{j}}\ ,
$$
where $\{ \ket{\phi_i} \}_{i\geq 0}$ is an orthonormal basis in the single particle Hilbert space with first element exactly the selected state corresponding to the complete set of properties.
From $\bkmv{\Psi}{\mathcal{E}_P}{\Psi}=1$ it follows that $\mathcal{E}_P \ket{\Psi}=\ket{\Psi}$,
which in turn can hold if and only if $c_{ij}=0$, $\forall i,j \neq 0$.
Moreover, particle identity implies $c_{0j}=\pm c_{j0}$, and normalization requires
$\modu{c_{00}}^2+2\sum_{j>0} \modu{c_{0j}}^2 =1$.
Thus,
\begin{eqnarray}
\nonumber
\ket{\Psi} &=& c_{00}\ket{\phi_{0}}\otimes\ket{\phi_0} + \ket{\phi_{0}}\otimes\left( \sum_{j>0} c_{0j} \ket{\phi_{j}}\right) + \left(\sum_{i>0} c_{i0}\ket{\phi_{i}}\right)\otimes \ket{\phi_0}\\
\label{eq:state}
&=& \ket{\phi_0}\otimes \left( \dfrac{c_{00}}{2}\ket{\phi_0} + \sum_{j>0} c_{0j} \ket{\phi_{j}} \right) +
\left( \dfrac{c_{00}}{2}\ket{\phi_0} + \sum_{j>0} c_{j0} \ket{\phi_{j}} \right)\otimes\ket{\phi_{0}}\ .
\end{eqnarray}
We now distinguish two cases.
\begin{itemize}
\item
Fermions: $c_{00}=0$ since there cannot be two Fermions in a same state. Then, $\ket{\Psi}$ is the anti-symmetrization of a tensor product vector state,
\begin{equation}
\label{eq:statefer}
\ket{\Psi}=\dfrac{1}{\sqrt{2}} \Big( \ket{\phi_{0}}\otimes\ket{\Upsilon} - \ket{\Upsilon}\otimes\ket{\phi_{0}} \Big)\ ,\quad
\ket{\Upsilon} = \sqrt{2}\sum_{j>0} c_{0j} \ket{\phi_{j}} \ ,
\end{equation}
where $\ket{\Upsilon}$ is orthogonal to $\ket{\phi_0}$.
\item
Bosons: using \eqref{eq:state}, $\ket{\Psi}$ results from the symmetrization of a tensor product vector state,
\begin{eqnarray}
\label{eq:statebos}
\ket{\Psi} &=& \sqrt{\dfrac{2-\modu{c_{00}}^2}{4}} \Big( \ket{\phi_{0}}\otimes\ket{\Theta} + \ket{\Theta}\otimes\ket{\phi_{0}} \Big)\\
\label{eq:statebos1}
\ket{\Theta} &=& \sqrt{\dfrac{4}{2-\modu{c_{00}}^2}} \left( \dfrac{c_{00}}{2}\ket{\phi_0} + \sum_{j>0} c_{j0} \ket{\phi_{j}} \right) \ .
\end{eqnarray}
\end{itemize}
\qed
\end{proof}

The previous proposition states that the symmetrization or anti-symmetrization of two-particle tensor product states do not forbid the attribution of a complete set of properties to one of the two parties, but only that we cannot establish by which one is the complete set of properties possessed.
This is not yet an assertion of separability of the state $\ket{\Psi}$ which, according to definition \ref{gs}, must be such that $\bkmv{\Psi}{\mathcal{E}_{Q}}{\Psi} =1$ for another projection $\mathcal{E}_Q$ depending on a one-dimensional single-particle projection $Q$, in general different from $P$: only in this case both parties can be declared to possess a complete set of properties, without, of course, being possible to establish which one has which property.
\bigskip

\begin{remark}
\label{rem1}
The latter statement requires that the single-particle one-dimensional projections $P$ and $Q$ be orthogonal,
for $P\neq Q$. In fact, only in this case, the product
$\mathcal{E}_P\mathcal{E}_Q$ is also a projection, equal to  $P\otimes Q+Q\otimes P$.
\end{remark}

\begin{itemize}
\item
Fermions: the structure in \eqref{eq:statefer} identifies $P$, respectively $Q$, as $P=\elr{\phi_0}{\phi_0}$, respectively $Q=\elr{\Upsilon}{\Upsilon}$.
It follows that, relative to the state $\ket{\Psi}$, one Fermion possesses the complete set of properties associated with $P$ and the other one the complete set of properties associated with the projection $Q$ orthogonal to $P$, so that $\ket{\Psi}$ is separable.
\begin{coro}
A vector state $\ket{\Psi}$ of two Fermions is separable if and only if $\ket{\Psi}$ is obtained by anti-symmetrizing a tensor product of two orthogonal single-particle vector states.
\end{coro}
\item
Bosons: according to Remark \ref{rem1} and to the structure of the state in \eqref{eq:statebos}, three cases have to be distinguished:
\begin{enumerate}
\item
$c_{j0}=0 \ \forall j>0$: then, $\ket{\Psi} = \ket{\phi_{0}}\otimes\ket{\phi_{0}}$. Both Bosons posses the complete set of properties associated with the single-particle projection $P=\elr{\phi_{0}}{\phi_{0}}$ and $\ket{\Psi}$ is separable.
\item
$c_{00}=0$: then, $\bk{\Theta}{\Phi_{0}}=0$. As for Fermions, the orthogonal projections $P=\elr{\phi_{0}}{\phi_{0}}$ and $Q=\elr{\Theta}{\Theta}$
represent two complete sets of properties possessed by the two Bosons, so that $\ket{\Psi}$ is separable.
\item
$\bk{\Theta}{\phi_{0}} \neq 0$ and $\ket{\Theta}\neq\ket{\phi_{0}}$: while it is legitimate to attribute to one of the two Bosons the complete set of properties associated either with the projection $P=\kb{\phi_0}{\phi_0}$ or $Q=\kb{\Theta}{\Theta}$,
one cannot attribute simultaneously the two properties to the two Bosons.
Indeed, no projection $Q\perp P$ exists such that $\bkmv{\Psi}{\mathcal{E}_P\mathcal{E}_Q}{\Psi}= 1$, hence $\ket{\Psi}$ is entangled.
\end{enumerate}
\begin{coro}
A vector state $\ket{\Psi}$  of two identical bosons is separable if and only if it is obtained by symmetrizing the tensor product of two orthogonal single-particle states or if it is the tensor product of a same single-particle state.
\end{coro}
\end{itemize}

\begin{remark}
\label{rem2}
Notice that, in the case of two Fermions, the attribution of one complete set of properties implies the simultaneous attribution of another orthogonal complete set of properties. Therefore, the attribution of one complete set of properties to a two fermion vector state $\ket{\Psi}$ is possible if and only if the state is separable. Instead, in the case of two Bosons, Case $3$ shows that there are vector states for which  two non-orthogonal complete sets of properties can be separately, but not simultaneously, attributed: these states are declared entangled.
\end{remark}

\subsection{Second Quantization}
\label{sec2.2}

Particle identity is at the core  of quantum statistical mechanics and the standard approach to such systems is in terms of creation and annihilation operators whose canonical commutation or anti-commutation relations embody the Bosonic or Fermionic character of the particles.
As for the representation Hilbert space, this is generated by acting with (powers of) the creation operators on a reference state $\ket{0}$ (called vacuum): typically, given the single particle Hilbert space $\mathbb{H}$, to any vector $\ket{\psi}\in\mathbb{H}$ one associates a creation operator which creates it by acting on the vacuum, $\hat{a}^\dag(\psi)\ket{0}=\ket{\psi}$, and an annihilation operator $\hat{a}(\psi)$ that destroys such a state, $a(\psi)\ket{\psi}=\ket{0}$.
Given two states $\ket{\psi_{1,2}}\in\mathbb{H}$, by expanding them with respect to an orthonormal basis $\{\ket{\phi_j}\}_j\in\mathbb{H}$ in the single-particle Hilbert space whose vectors are generated and destroyed by $\hat{a}^\dag_j\,,\,\hat{a}_j$,
the CCR and CAR given by \eqref{CCR} and \eqref{CAR} generalize to
\begin{equation}
\label{CCAR}
[\hat{a}(\psi_1)\,,\,\hat{a}^\dag(\psi_2)]=\bk{\psi_1}{\psi_2}=\{\hat{a}(\psi_1)\,,\,\hat{a}^\dag(\psi_2)\}\ .
\end{equation}
Orthonormal basis vectors are obtained by acting repeatedly with creation operators on the vacuum
\begin{equation}
\label{basis}
\ket{n_1,n_2,\ldots,n_k}=\frac{(\hat{a}_1^\dag)^{n_1}(\hat{a}_2^\dag)^{n_2}\cdots (\hat{a}_k^\dag)^{n_k} }{\sqrt{\prod_{j=1}^k n_j!}}\ket{0}\ .
\end{equation}
Each pair $\hat{a}_i\,,\,\hat{a}^\dag_i$ is associated to a possible mode of the system of identical particles so that these states  contain $n_i$ particles in the first $i$-th modes. Indeed, from the \eqref{CCR} and \eqref{CAR} follows that these states are eigenstates of the number operator
\begin{equation}
\label{Number}
\hat{N}=\sum_j\hat{a}_j^\dag \hat{a}_j\ ,\quad \hat{N}\ket{n_1,n_2,\ldots,n_k}=\Big(\sum_{j=1}^k n_j\Big)\,\ket{n_1,n_2,\ldots,n_k}\ .
\end{equation}
The vectors (\ref{basis}) form an orthonormal basis for the Hilbert space of $N=\sum_{j=1}^k n_j$ particles, with $n_j=0,1$ for Fermions. By varying $N$ one constructs the Fock Hilbert space $\mathcal{H}$ as the orthogonal sum of the sectors with finite numbers of Bosons or Fermions.

It is thus clear that, in a second quantized formalism, the building blocks are the annihilation and creation pairs $\hat{a}_i\,,\,\hat{a}_i^\dag$, together with the CCR and CAR, and their polynomials: we shall denote by $\mathcal{A}$ the Bose or Fermi algebra arising from such polynomials by closing them in norm in the Fermi case or via the generalized Weyl operators in the Bose case (see footnote \eqref{footnote1}).

\begin{remark}
\label{rem3}
The symmetric and anti-symmetric character of Bosonic and Fermionic vector states typical of the first quantization formalism appears when multi-particle states are represented with respect to the chosen basis. For instance, setting $\psi(i)=\bk{i}{\psi}$,
$\ket{\psi_1\psi_2}=\hat{a}^\dag(\psi_1)\hat{a}^\dag(\psi_2)\ket{0}$ and
$\ket{ij}=\hat{a}^\dag_i\hat{a}^\dag_j\ket{0}$, from \eqref{CCAR} one gets
$$
\bk{ij}{\psi_1\psi_2}=\bkmv{0}{\hat{a}_i\hat{a}_ja^\dag(\psi_1)\hat{a}^\dag(\psi_2)}{0}=\psi_1(i)\psi_2(j)\pm\psi_1(j)\psi_2(i)\ .
$$
Furthermore, when $i\neq j$, a basis vector state as $\ket{ij}$ corresponds in first quantization to
$$
\frac{\ket {i}\otimes\ket{j}\pm\ket{j}\otimes\ket{i}}{\sqrt{2}}\ .
$$
\end{remark}

The lack of a tensor product structure either in the algebra $\mathcal{A}$ or in the Hilbert space spanned by vectors as in \eqref{basis} requires the notions of locality of observables and separability of states to be defined within a purely algebraic context. The notion of local observables or of a local subalgebra of observables is typical of quantum statistical mechanics where one considers single particle states $\ket{\psi_V}$ supported within finite volumes $V\subset\mathbb{R}^3$ and the algebras $\mathcal{A}_V\subset\mathcal{A}$ generated by all annihilation and creation operators $\hat{a}(\psi_V)\,,\,\hat{a}^\dag(\psi_V)$.

Let us first consider the case of a Bosonic system.
If two disjoint volumes are considered, $V_1\cap V_2=\emptyset$, then, from the CCR one gets
$$
\Big[\hat{a}(\psi_{V_1})\,,\,\hat{a}^\dag(\psi_{V_2})\Big]=\langle\psi_{V_1}\vert\psi_{V_2}\rangle=\int_{\RI^3}{\rm d}^3r\,
\psi^*_{V_1}(r)\, \psi_{V_2}(r)=0\ .
$$
This commutativity extends to monomials and polynomials so that the sub-algebras $\mathcal{A}_{V_{1,2}}$ commute, namely
\begin{equation}
\label{commsub}
\Big[A_1\,,\,A_2\Big]=0\qquad \forall A_1\in\mathcal{A}_{V_1}\ ,\ A_2\in\mathcal{A}_{V_2}\ .
\end{equation}
This commutativity corresponds to algebraic independence as it implies that measuring an observable in $\mathcal{A}_{V_1}$ cannot influence the simultaneous measurement of an observable in $\mathcal{A}_{V_2}$.
As already mentioned in the introduction, in \cite{benffm2010} algebraic independence is used to consistently generalize the standard notion of locality connected with algebraic tensor product structures as $M_2(\mathbb{C})\otimes M_2(\mathbb{C})$ for two
distinguishable  qubits.

\begin{remark}
\label{rem4}
The case of Fermions requires particular care as in the same conditions as before, we know that
$$
\Big\{\hat{a}(\psi_{V_1})\,,\,\hat{a}^\dag(\psi_{V_2})\Big\}=\langle\psi_{V_1}\vert\psi_{V_2}\rangle=\int_{\RI^3}{\rm d}^3r\,
\psi^*_{V_1}(r)\, \psi_{V_2}(r)=0\ ,
$$
but nothing can be said about the fate of commutators $\Big[\hat{a}(\psi_{V_1})\,,\,\hat{a}^\dag(\psi_{V_2})\Big]$.
In order to be sure that two Fermionic operators  be commuting, at least one of them must be within the algebra generated by polynomials of even order.
Indeed, one can then use \eqref{CAR-CCR} and the anti-commutativity of $\mathcal{A}_{V_{1,2}}$ to see that they commute.
Therefore,  one can meaningfully ask about the commutativity of the  pair of sub-algebras $(\mathcal{A}^{ev}_{V_1}\,,\,\mathcal{A}^{ev}_{V_2})$,
$(\mathcal{A}^{odd}_{V_1}\,,\,\mathcal{A}^{ev}_{V_2})$ and $(\mathcal{A}^{ev}_{V_1}\,,\,\mathcal{A}^{odd}_{V_2})$, where
$\mathcal{A}_V^{ev,odd}$ are the sub-algebras generated by even, respectively odd polynomials in the Fermionic annihilation and creation operators of functions localized within the volume $V$.
\end{remark}

The pairs $(\mathcal{A}_1,\mathcal{A}_2)$ of sub-algebras we consider in what follows will be assumed to be generated by two sets of creation and annihilation operators $\{(\hat{a}_i,\hat{a}^\dag_i): i\in I_{1,2}\}$.
We can now formulate the notion of locality in purely algebraic terms.

\begin{mydef}\hfill
\label{defloc}

Any pair $\left( \mathcal{A}_{1},\mathcal{A}_{2} \right)$ of sub-algebras $\mathcal{A}_{1,2}\subseteq\mathcal{A}$ such that $\mathcal{A}_1 \cup \mathcal{A}_1=\mathcal{A}$ will be called a bipartition if
\begin{enumerate}
\item
Bosonic case: $\mathcal{A}_1$ and $\mathcal{A}_2$ commute;
\item
Fermionic case: $\mathcal{A}^{ev}_1$ and $\mathcal{A}^{ev}_2$ commute, as well as $\mathcal{A}^{ev}_1$ and $\mathcal{A}^{odd}_2$, $\mathcal{A}^{odd}_1$ and $\mathcal{A}^{ev}_2$.
\end{enumerate}
Furthermore,  $A\in\mathcal{A}$ will be called local with respect to
$\left( \mathcal{A}_{1},\mathcal{A}_{2} \right)$ if
\begin{enumerate}
\item
Bosonic case: $A=A_{1}A_{2}$ with $A_{1} \in \mathcal{A}_{1}$ and $A_{2} \in \mathcal{A}_{2}$;
\item
Fermionic case:  $A=A_{1}A_{2}$ with either $A_{1} \in \mathcal{A}^{ev}_{1}$ and $A_{2} \in \mathcal{A}^{ev}_{2}$, or
$A_{1} \in \mathcal{A}^{ev}_{1}$ and $A_{2} \in \mathcal{A}^{odd}_{2}$, or $A_{1} \in \mathcal{A}^{odd}_{1}$ and $A_{2} \in \mathcal{A}^{ev}_{2}$.
\end{enumerate}
\end{mydef}

\begin{remark}
\label{rem5}
As in the case of distinguishable particles where one may consider multi-partite entanglement with respect to multiple tensor product algebras, for instance
$M_2(\mathbb{C})^{\otimes n}$ in the case of $n$ qubits, also in the second quantized approach one could consider a multi-partition of $\mathcal{A}$ into a set of $n$ mutually commuting sub-algebras that generate it.
However, in the following we shall stick to the case of two commuting sub-algebras $\mathcal{A}_{1,2}\subset\mathcal{A}$, exactly as, for standard qubits, it occurs with  the bipartition $\Big(M_2(\mathbb{C})\otimes
\mathbb{I},\mathbb{I}\otimes M_2(\mathbb{C})\Big)$ which generates $M_4(\mathbb{C})$.
\end{remark}

One of the advantages of the algebraic approach is that the notion of state can be extended beyond Hilbert state vectors and density matrices to that
of positive and normalized functionals $\omega$ over the algebra $\mathcal{A}$, namely linear maps from $\mathcal{A}$ into $\mathbb{C}$ such that
\begin{equation}
\label{plf}
\omega(A^\dag\,A)\geq 0\qquad\forall\, A\in\mathcal{A}\quad\hbox{and}\quad \omega(\mathbb{I})=1\ ,
\end{equation}
where $\omega(A)$ is called the expectation of $A$ with respect to $\omega$ \cite{BraRob}.

In the algebraic approach, mixed states on the algebra $\mathcal{A}$ are those $\omega$ that can be written as convex combinations of other functionals on $\mathcal{A}$, $\omega=\sum_j\lambda_j\,\omega_j$, $\lambda_j\geq 0$ and $\sum_j\lambda_j=1$. If a state $\omega$ cannot be convexly decomposed, it is called pure.
On finite dimensional Hilbert spaces, these positive, normalized functionals $\omega$ boil down to either pure state projections or to density matrices; that is, if $\mathcal{A}=M_n(\mathbb{C})$ then, for all $A\in M_n(\mathbb{C})$,
\begin{equation}
\label{plffd}
\omega(A)=\langle\psi\vert A\vert\psi\rangle={\rm Tr}\Big(\kb{\psi}{\psi}\,A\Big)\ ,\quad\hbox{or}\quad
\omega(A)={\rm Tr}\Big(\rho\,A\Big)\ ,
\end{equation}
where $\rho\in M_n(\mathbb{C})$ is a density matrix.

In the case of two distinguishable qubits, a density matrix
$\rho\in M_4(\mathbb{C})$ is separable when $\rho$ is a convex sum of tensor products of single particle density matrices,
$\rho=\sum_j\lambda_j\,\rho^{(1)}_j\otimes\rho^{(2)}_j$, $\lambda_j\geq 0$ and $\sum_j\lambda_j=1$ so that
$$
{\rm Tr}\Big(\rho\,A_1\otimes A_2\Big)=\sum_j\lambda_j\,{\rm Tr}\Big(\rho^{(1)}_j\,A_1\Big)\,{\rm Tr}\Big(\rho^{(2)}_j\,A_2\Big)
$$
on all local operators $A_1\otimes A_2$.

Together with the previous definition of bipartitions, this leads to the following generalization of the notion of separable states.

\begin{mydef}
\label{def:ben}
A state $\omega$ on the algebra $\mathcal{A}$  is separable with respect to the bipartition $\left( \mathcal{A}_{1},\mathcal{A}_{2} \right)$, or $\left( \mathcal{A}_{1},\mathcal{A}_{2} \right)$-separable,
if, for all $\left( \mathcal{A}_{1},\mathcal{A}_{2} \right)$-local operators $A=A_1A_2$,
\begin{equation}
\label{sep2q}
\omega(A_{1}A_{2}) = \sum_{k} \lambda_{k}\ \omega_{k}^{(1)}(A_{1})\, \omega_{k}^{(2)}(A_{2}) \qquad\quad \lambda_{k} \geq 0, \qquad \sum_{k} \lambda_{k}=1 \ ,
\end{equation}
where $\omega$, $\omega_{k}^{(1)}$ and $\omega_{k}^{(2)}$ are other states on $\mathcal{A}$.
Otherwise, the state $\omega$ is said to be $\left( \mathcal{A}_{1},\mathcal{A}_{2} \right)$-entangled.
\end{mydef}

In the second quantized formalism, the possibility of uncountable many pairs of local sub-algebras and thus of  states that can be separable with respect to a bi-partition and entangled with respect to another one is evident.
In the case of distinguishable particles, this fact is masked by the natural distinction of particles, this in turn being related to the
a-priori  tensor product structure of the Hilbert space.
As an example, consider two qubits: states as $\ket{ij}=\ket{i}\otimes\ket{j}$, with $\ket{i}$, $i=1,2$, an orthonormal basis in $\mathbb{C}^2$, are separable with respect to the tensor product structure $\mathbb{C}^2\otimes\mathbb{C}^2$.
However, through the Bell states
$$
\ket{\psi_\pm}=\frac{\ket{00}\pm\ket{11}}{\sqrt{2}}\ ,\quad \ket{\phi_\pm}=\frac{\ket{01}\pm\ket{10}}{\sqrt{2}}\ ,
$$
that are maximally entangled with respect to the bipartition $\Big(M_2(\mathbb{C})\otimes \mathbb{I},\mathbb{I}\otimes M_2(\mathbb{C})\Big)$, and the $4\times 4$ matrices
$$
\kb{\psi_+}{\psi_+}\ ,\ \kb{\psi_-}{\psi_+}\ ,\ \kb{\psi_+}{\psi_-}\ ,\ \kb{\psi_-}{\psi_-}\ ,
$$
respectively
$$
\kb{\phi_+}{\phi_+}\ ,\ \kb{\phi_-}{\phi_+}\ ,\ \kb{\phi_+}{\phi_-}\ ,\ \kb{\phi_-}{\phi_-}\ ,
$$
one constructs two sub-algebras $\mathcal{A}_\psi$ and $\mathcal{A}_\phi$. These are isomorphic to $M_2(\mathbb{C})$, commute with each other because of the orthogonality of the Bell states and generate an algebra isomorphic to $M_4(\CI)$.
With respect to  the bipartition $(\mathcal{A}_{\psi},\mathcal{A}_\phi)$, the states as $\ket{ij}$ are entangled, while the Bell states are separable.

Differently from the first quantization approach where the attribution of properties work only for vector states, the algebraic approach based on the second quantization formalism  provides tools that are valid for pure and mixed states. For instance, for $N$ Bosons and two mode one can prove that the partial transposition criterion is an exhaustive entanglement witness exactly as for two qubits or one qubit and one qutrit \cite{benffm2010}.

In the following, we shall focus upon Bosons and Fermions having at disposal a number $M$ of modes, possibly infinite, associated with annihilation and creation operators $\hat{a}_i\,,\,\hat{a}^\dag_i$, $1\leq i\leq M$.
We shall construct a bipartition of the algebra $\mathcal{A}$ by fixing an integer $m$, $0 \leq m \leq M$, and by considering  the subsets
$$
\Big\{\cop_{i},\aop_{i}\ :\  i=1,2,\dots,m \Big\}\ ,\quad \Big\{ \cop_j,\aop_j\ :\ j=m+1,m+2,\dots,M \Big\}\ .
$$
We shall denote by $\mathcal{A}_1\subseteq\mathcal{A}$ the subalgebra generated by the first set, by $\mathcal{A}_2$ that generated by the second one.
Furthermore, in order to be able to compare the two approaches, we shall deal with pure states, only; we first review and then extend to Fermions a result that was previously proved to hold for Bosons \cite{benffm2012}.

\subsubsection{Separable Bosonic and Fermionic Pure States}

In the standard setting, normalized separable vector states $\ket{\psi}=\ket{\psi_1}\otimes\ket {\psi_2}$
are such that
$$
\langle\psi\vert A_1\otimes A_2\vert\psi\rangle=\langle\psi_1\vert A_1\vert\psi_1\rangle\,\langle\psi_2\vert A_2\vert \psi_2\rangle
=\langle\psi\vert A_1\otimes \mathbb{I}\vert\psi\rangle\,\langle\psi\vert\mathbb{I}\otimes A_2\vert\psi\rangle\ ,
$$
for all operators $A_{1,2}$ acting on their respective Hilbert spaces.
An analogous relation holds in the algebraic setting.

\begin{lemma}
\label{lemma1}
Bosonic vector states states $\omega(X)=\bkmv{\Psi}{X}{\Psi}$ on $\mathcal{A}$, where $\ket{\Psi}$ belongs to the Fock Hilbert space $\mathcal{H}$ generated by vectors as in \eqref{basis}, are separable with respect to a bipartition $(\mathcal{A}_1,\mathcal{A}_2)$ if and only if
\begin{equation}
\label{pureom}
\bkmv{\Psi}{A_1\,A_2}{\Psi} = \bkmv{\Psi}{A_1}{\Psi}\, \bkmv{\Psi}{A_2}{\Psi}\qquad\forall\, A_1\in\mathcal{A}_1\,,\,A_2\in\mathcal{A}_2\ .
\end{equation}
The same factorization characterizes Fermionic vector states that remain pure on the even sub-algebra generated by $\mathcal{A}^{ev}_{1,2}$.
\end{lemma}

\begin{proof}
According to Definition \ref{def:ben},  states satisfying \eqref{pureom} are automatically $\left( \mathcal{A}_{1},\mathcal{A}_{2} \right)$-separable as they acts as in  \eqref{sep2q} with only one convex contribution in factorized form.

In the Bosonic case,  $\left( \mathcal{A}_{1},\mathcal{A}_{2} \right)$-local operators generate the whole algebra $\mathcal{A}$; then, any $A\in\mathcal{A}$ can be written as  $A=\sum_{a,b} C_{ab}\, A^{(1)}_a\,A^{(2)}_b$ with $A^{(1)}_a\in\mathcal{A}_1$ and $A^{(2)}_b\in\mathcal{A}_2$.
Therefore, if a vector state satisfies \eqref{sep2q} on all $\left(\mathcal{A}_{1},\mathcal{A}_{2}\right)$-local operators, then
$$
\bkmv{\Psi}{A}{\Psi}=\sum_k\lambda_k\sum_{a,b}C_{ab}\,\omega^{(1)}_k(A_a^{(1)})\,\omega^{(2)}_k(A_b^{(2)})=\sum_k\lambda_k\tilde{\omega}_k(A)\ ,
$$
in terms of other states $\tilde{\omega}_k$  defined on the whole of $\mathcal{A}$ 
by means of $\tilde{\omega}_k(A)=$\break $\sum_{a,b} C_{ab}\, \omega_k^{(1)}(A^{(1)}_a)\omega_k^{(2)}(A^{(2)}_b)$.
Therefore, $\ket{\Psi}$ can correspond to a pure separable state $\omega$ only if $\lambda_k\neq 0$ for just a single $k$ whence \eqref{pureom} is satisfied. Indeed,
$\bkmv{\Psi}{A_1A_2}{\Psi}=\omega^{(1)}_k(A_1)\,\omega^{(2)}_k(A_2)$ yields $\omega^{(1,2)}_k(A_{1,2})=\bkmv{\Psi}{A_{1,2}}{\Psi}$ for $A_1=\mathbb{I}$ and $A_2=\mathbb{I}$.

Because of the assumed purity of $\ket{\Psi}$ on the sub-algebra generated by the even components of the bipartition
$(\mathcal{A}_{1},\mathcal{A}_2)$, the same argument  holds for Fermions.
\qed
\end{proof}

\begin{remark}
\label{remfond}
There is a fundamental difference between Bosonic and Fermionic systems: in the latter case, as seen in Remark \ref{rem4}, there are two odds components in the bipartition $(\mathcal{A}_1,\mathcal{A}_2)$ and operators in these components anti-commute.
In Definitions \ref{defloc} and \ref{def:ben} 
the notions of locality and of state separability have been
formulated in terms of commuting observables only.
Therefore, it may happen that vector states as $\ket{\Psi}$, which are pure on the whole algebra $\mathcal{A}$ generated by $\mathcal{A}_{1,2}$,  turn out to be no longer pure when restricted to the Fermionic sub-algebra generated by $\mathcal{A}_{1,2}^{ev}$.
This possibility will not be considered in the following since it is not needed for the subsequent discussion 
\end{remark}

The next proposition gives the explicit expression of the pure states fulfilling the requests of the previous lemma.

\begin{proposition}
\label{thm:main}
A normalized Bosonic vector state $\ket{\Psi}$ in the  Fock Hilbert space $\mathcal{H}$ spanned by vectors as in \eqref{basis} is $\left(\mathcal{A}_{1},\mathcal{A}_{2} \right)$-separable if and only if it is generated by a
$\left(\mathcal{A}_{1},\mathcal{A}_{2} \right)$-local operator,
\begin{equation}
\label{eq:thmben}
\ket{\Psi} = \mathcal{P}(\hat{a}^\dag_{1},\dots,\hat{a}^\dag_{m}) \cdot \mathcal{Q}(\hat{a}^\dag_{m+1},\dots,\hat{a}^\dag_{M}) \ket{0}\ ,
\end{equation}
where $\mathcal{P}$, $\mathcal{Q}$ are polynomials in the creation operators relative to the first $m$ modes and the last $M-m$ modes, respectively. Otherwise, the state is entangled.

The same is true of  normalized Fermionic vector states $\ket{\Psi}$ that remain  pure when restricted to the even sub-algebra generated by $\mathcal{A}_{1,2}^{ev}$.
\end{proposition}

\begin{proof}
Consider first the boson case. Vector states as in the statement of the proposition have the general form
\begin{equation}
\label{aid1}
\ket{\Psi}=\sum_{\{k\},\{\alpha\}} C_{\{k\},\{\alpha\}} \ket{k_{1},\dots,k_{m};\alpha_{m+1},\dots,\alpha_{M}} \, , \qquad \sum_{\{k\},\{\alpha\}} \modu{C_{\{k\},\{\alpha\}}}^2=1\ ,
\end{equation}
where $\{k\}=(k_1,k_2,\ldots,k_m)$, respectively $\{\alpha\}=(\alpha_{m+1},\alpha_{m+2},\ldots,\alpha_{M})$, is the vector of occupation numbers of the first $m$, respectively second $M-m$ modes and
$$
\ket{\{k\},\{\alpha\}}=\ket{k_{1},\dots,k_{m};\alpha_{m+1},\dots,\alpha_{M}}=\frac{(\hat{a}_1^\dag)^{k_1}\cdots(\hat{a}_m^\dag)^{k_m}(\hat{a}_{m+1}^\dag)^{\alpha_{m+1}}\cdots
(\hat{a}_M^\dag)^{\alpha_M}}{\sqrt{k_1!\cdots k_m!\alpha_{m+1}!\cdots\alpha_M!}}\ket{0}\ .
$$
From Lemma \ref{lemma1}, the separability of $\ket{\Psi}$ is equivalent to
$\langle\Psi\vert A_1A_2\vert\Psi\rangle=\langle\Psi\vert A_1\vert\Psi\rangle\,\langle\Psi\vert A_2\vert\Psi\rangle$ for all choices of $A_{1,2}$ in the commuting subalgebras $\mathcal{A}_{1,2}$.
We should then use this request to derive relations among the coefficients $C_{\{k\},\{\alpha\}}$ that force them to factorize: $C_{\{k\},\{\alpha\}}=C^{(1)}_{\{k\}}\,C^{(2)}_{\{\alpha\}}$.

To this purpose, in a first quantization context, the obvious operators $A_{1,2}$ to be used would be
$A_1=\kb{\{p'\}}{\{p\}}$ and $A_2=\kb{\{\beta'\}}{\{\beta\}}$ and
$A_1 A_2=\kb{\{p'\},\{\beta'\}}{\{p\},\{\beta\}}$.
In the second quantization setting, these operators are replaced by
\begin{eqnarray}
\label{aid2}
&&A_{1} = (\hat{a}^\dag_{1})^{p'_{1}}\dots(\hat{a}^\dag_{m})^{p'_{m}}\,
\left(\dfrac{1}{2\pi i} \oint_{\Gamma(0)} \dfrac{dz}{z-\hat{N}_{1}}\right)\, \hat{a}_{1}^{p_{1}}\dots \hat{a}_{m}^{p_{m}} \\
\label{aid3}
&&A_{2} = (\hat{a}^\dag_{m+1})^{\beta'_{m+1}}\dots(\hat{a}^\dag_{M})^{\beta'_{M}}\,
\left(\dfrac{1}{2\pi i} \oint_{\Gamma(0)} \dfrac{dz}{z-\hat{N}_{2}}\right)\, \hat{a}_{m+1}^{\beta_{m+1}}\dots \hat{a}_{M}^{\beta_{M}}\\
\nonumber
&&A_1A_2 = (\hat{a}^\dag_{1})^{p'_{1}}\dots(\hat{a}^\dag_{m})^{p'_{m}}\,(\hat{a}^\dag_{m+1})^{\beta'_{m+1}}\dots(\hat{a}^\dag_{M})^{\beta'_{M}}\,
\left(\dfrac{1}{2\pi i} \oint_{\Gamma(0)} \dfrac{dz}{z-\hat{N}}\right)\\
\label{aid4}
&&\hskip 7cm \times
\hat{a}_{m+1}^{\beta_{m+1}}\dots \hat{a}_{m+M}^{\beta_{m+M}}\, \hat{a}_{1}^{p_1}\dots \hat{a}_{m}^{p_m}\ ,
\end{eqnarray}
where $p_i,p'_i$ and $\beta_j,\beta'_j$ are  integers, while  $\hat{N}_1=\sum_{k=1}^m\hat{a}^\dag_k\,\hat{a}_k$, $\hat{N}_2=\sum_{j=m+1}^M \hat{a}_j^\dag\,\hat{a}_j$ and $\hat{N}=\hat{N}_1+\hat{N}_2$ are the number operators relative to the two sub-sets of modes and their union, while
$\Gamma(0)$ is a contour around $z=0$ excluding all other integers.
The choice of contour forces the three integrals to vanish unless $z=0$, whence the first two project onto the sub-spaces with no particles in the corresponding sub-sets of modes and the third one onto the vacuum.

Then one calculates
\begin{eqnarray}
\label{aid5}
\bkmv{\Psi}{A_1}{\Psi} &=&
\left( \prod_{i=1}^{m} \sqrt{p_{i}! p'_{i}!} \right) \sum_{\{\alpha\}} \overline{C}_{\{p'\},\{\alpha\}} C_{\{p\},\{\alpha\}}\\
\label{aid6}
\bkmv{\Psi}{A_2}{\Psi} &=&
\left( \prod_{j=m+1}^{M} \sqrt{\beta_{j}! \beta'_{j}!}\right) \sum_{\{k\}} \overline{C}_{\{k\},\{\beta'\}} C_{\{k\},\{\beta\}}\\
\label{aid7}
\bkmv{\Psi}{A_1A_2}{\Psi} &=&
\left( \prod_{i=1}^{m} \sqrt{p_{i}! p'_{i}!} \right)\left( \prod_{j=m+1}^{M} \sqrt{\beta_{j}! \beta'_{j}!}\right) \overline{C}_{\{p'\},\{\beta'\}} C_{\{p\},\{\beta\}}\ ,
\end{eqnarray}
so that
\begin{equation}
\label{constcon}
\overline{C}_{\{p'\},\{\beta'\}} C_{\{p\},\{\beta\}}  =  \left( \sum_{\{\alpha\}} \overline{C}_{\{p'\},\{\alpha\}} C_{\{p\},\{\alpha\}} \right) \left( \sum_{\{k\}} \overline{C}_{\{k\},\{\beta'\}} C_{\{k\},\{\beta\}} \right)\ .
\end{equation}
For $p'=p$ and $\beta'=\beta$ this expression becomes
$$
\modu{C_{\{p\},\{\beta\}}}^2 =  \left( \sum_{\{\alpha\}} \modu{C_{\{p\},\{\alpha\}}}^2 \right)\left( \sum_{\{k\}} \modu{C_{\{k\},\{\beta\}}}^2 \right)\ .
$$
Setting $D_{\{ p \}}=\sum_{\{ \alpha \}} \modu{C_{\{p\},\{\alpha\}}}^2$ and $D'_{\{ \beta \}}=\sum_{\{ k \}} \modu{C_{\{k\},\{\beta\}}}^2$, one rewrites
\begin{equation}
\label{aid8}
C_{\{p\},\{\beta\}} = \sqrt{D_{\{p\}}}\, \sqrt{D'_{\{\beta\}}}\, \e^{i \theta_{\{p\}\{\beta\}}}\ .
\end{equation}
Inserting this expression in \eqref{constcon}, we obtain
$$
\e^{i(\theta_{\{p'\}\{\beta'\})}-\theta_{\{p\}\{\beta\}})}= \sum_{\{ \alpha \}} D'_{\{\alpha\}} \e^{i \left( \theta_{\{p\}\{\alpha\}}-\theta_{\{p'\}\{\alpha\}}\right)}\times \sum_{\{k\}} D_{\{k\}}
\e^{i \left( \theta_{\{k\}\{\beta\}}-\theta_{\{k\}\{\beta'\}}\right)}\ .
$$
Since $\sum_{\{p\}} D_{\{ p \}}=1=\sum_{\{\beta\}} D'_{\{ \beta \}}$, by setting $\beta'=\beta$
one sees that $\theta_{\{p\}\{\beta\}}- \theta_{\{p'\}\{\beta'\}}=\phi_{pp'}$ for all $\beta$.
Fixing an arbitrary $p'$ and inserting this expression into \eqref{aid8} yields
\begin{equation}
\label{aid9}
C_{\{p\},\{\beta\}} = \sqrt{D_{\{p\}}}\,\e^{i \phi_{pp'}} \times \sqrt{D'_{\{\beta\}}}\, \e^{i\theta_{\{p'\}\{\beta\}}}\ .
\end{equation}

In the case of Fermions, the integers $p_i,p'_j$, $\beta_i,\beta'_j$ in (\ref{aid2})-(\ref{aid4}), can not be arbitrarily chosen;
the two subalgebras must be even so that $\sum_{i=1}^m(p_i+p'_i)=2r$ and $\sum_{i=m+1}^M(\beta_i+\beta'_i)=2s$ for suitable integers $r,s$.
Furthermore, since the number operators are sums of quadratic monomials, by series expansion the  integrals appearing
in those equations provide operators that are elements of the even sub-algebras $\mathcal{A}^{ev}_{1,2}$ and $\mathcal{A}^{ev}$.
Therefore, relations \eqref{aid2}--\eqref{aid7} hold true also for Fermions. As in the Boson case, the proof is completed by restricting to $\{p\}=\{p'\}$ and $\{\beta\}=\{\beta'\}$ which guarantees that the even condition is fulfilled,
whence the argument also applies to Fermions.\qed
\end{proof}

\begin{remark}
\label{rem6}
In the proof of the previous proposition nothing depended on having a finite number $m$ of modes in the first set and a finite number $M$ of modes in the second set.
The result thus extends to the case of infinite disjoint sets of modes for all $\ket{\Psi}$ of unit norm.
\end{remark}

\section{First and Second Quantization Approaches Compared}
\label{sec3}

Taking advantage of the previous results, in this section we study the relations  between the first (particle) and second quantization (mode) approach to entangled vector states of identical particle systems.

\subsection{From First to Second Quantization}
\label{sec3.1}

The first quantization approach in \cite{ghigmw2002,GhiMar} deals with two identical particles in a vector state $\ket{\Psi}$ with an infinite dimensional
single particle Hilbert space $\mathbb{H}$. In that approach, vector states are separable if and only if they possess two orthogonal sets of complete properties associated with single particle orthogonal one-dimensional projections $P$ and $Q$.  Fix an orthonormal basis $\{\ket{\psi_j}\}_j$ in $\HI$ and let $\ket{\psi_1}$ be the single-particle vector state
corresponding to the attribution of a complete set of properties to one of the two Bosons: $P=\kb{\psi_1}{\psi_1}$.

In the second quantization approach, let $\hat{a}_j,\hat{a}^\dag_j$ be the annihilation and creation operators of the states $\ket{\psi_j}$; then, a general two-particle vector state can be written as
\begin{equation}
\label{comp0}
\ket{\Psi} = \Big( c_{11} \Big(\cop_{1}\Big)^2 + \cop_{1} \sum_{j>1} c_{1j} \cop_{j} + \sum_{i,j>1} c_{ij} \cop_{i} \cop_{j} \Big) \ket{0}\ .
\end{equation}
Furthermore, the attribution of a complete set of properties  naturally leads to consider the bipartition where $\mathcal{A}_1$ is generated by $\hat{a}_1\,,\,\hat{a}^\dag_1$ and $\mathcal{A}_2$ by the remaining modes:
\begin{equation}
\label{comp1}
\mathcal{A}_{1}= \Big\{ \cop_{1}, \aop_{1} \Big\}\ ,\quad \mathcal{A}_{2} = \Big\{ \cop_{j}, \aop_{j}\Big\}_{j\geq 2}\ .
\end{equation}
Also, the second quantized expression of the projection operator $\mathcal{E}_P$ in \eqref{symmprop2} reads
\begin{equation}
\label{comp2}
\mathcal{E}_{P} = \dfrac{1}{2} \cop_{1}\aop_{1} \left( 3 - \cop_{1}\aop_{1} \right)\ .
\end{equation}
Indeed, acting on the $2$ particle sector, the right hand side does not vanish only on vector states with at least one particle in the first mode, in which case they are eigenstates of the right hand side with eigenvalue $1$ .
Moreover, according to Definition \ref{defloc}, $\mathcal{E}_P$ is local with respect to the bipartition $(\mathcal{A}_1,\mathcal{A}_2)$.

According to Definition \ref{def:ben}, in order to be separable $\ket{\Psi}$ must surely satisfy $\bkmv{\Psi}{\mathcal{E}_{P}}{\Psi}=1$ whence
$\mathcal{E}_{P} \ket{\Psi} = \ket{\Psi}$. Then,
$$
\mathcal{E}_{P}\ket{\Psi} = \left( c_{11} \left(\cop_{1}\right)^2 + \cop_{1} \sum_{j>1} c_{1j} \cop_{j} \right) \ket{0}=\ket{\Psi}
$$
implies that $c_{ij}=0$ when both $i$ and $j$ are different from one, so that
$$
\ket{\Psi} = \cop_{1} \left( c_{11} \cop_{1} + \sum_{j>1} c_{1j} \cop_{j} \right) \ket{0}\ .
$$
We  now distinguish the Fermionic from the Bosonic case.

\subsubsection{Fermions}
\label{sec3.2}

Since there cannot be two Fermions in a same single particle state, $c_{11}=0$, and
$$
\ket{\Psi} = \cop_{1}\,\Big( \sum_{j>1} c_{1j} \cop_{j}\Big)\, \ket{0}=\mathcal{P}(\cop_{1}) \cdot \mathcal{Q}(\cop_{2},\ldots,\cop_{j}\ldots)\ket{0}\ .
$$
Therefore, the vector state $\ket{\Psi}$ can be recast in the form  \eqref{eq:thmben} by means of a monomial in $\hat{a}^\dag_1\in \mathcal{A}_1$  and of a first order polynomial in $\hat{a}_j^\dag\in\mathcal{A}_2$.
According to Proposition \ref{thm:main}, the vector state $\ket{\Psi}$  is thus $(\mathcal{A}_1,\mathcal{A}_2)$-separable, with $\Big( \sum_{j>1} c_{1j} \cop_{j}\Big)\, \ket{0}$ orthogonal to $\cop_{1}\, \ket{0}$.
Thus, every $2$-Fermion vector state which is separable according to the first quantization approach is also  $(\mathcal{A}_1,\mathcal{A}_2)$-separable.

\subsubsection{Bosons}

We have set $a^\dag_1\ket{0}=\ket{\psi_1}$; let $\displaystyle\ket{\Theta}$ denote the normalized vector state $\Big( c_{11} \cop_{1} + \sum_{j\geq2} c_{1j} \cop_{j} \Big) \ket{0}$.
\begin{enumerate}
\item
Case when, in the first quantization approach, $\ket{\Psi}$ is separable because both Bosons are in the same state. \\
In the second quantization approach, this amounts  to $c_{1j}=0$, for all $j\geq 2$. Then, $\ket{\Theta}=\ket{\psi_1}$ and $\displaystyle \ket{\Psi}=\frac{(a^\dag_1)^2}{\sqrt{2!}}\ket{0}$. This state  is as in \eqref{eq:thmben} in terms of a second order monomial in the first mode creation  operator and thus, according to Proposition \ref{thm:main},  $(\mathcal{A}_1,\mathcal{A}_2)$-separable.
\item
Case when, in the first quantization approach, $\ket{\Psi}$ is separable because both Bosons can be attributed a complete set of properties associated with orthogonal single particle vectors.\\
In the second quantization approach, this case is recovered with the choice $c_{11}=0$.
Then, $\bk{\Theta}{\psi_1} = 0$ and one complete set of properties corresponds to $\mathcal{E}_P$ in  \eqref{comp2}, the other one to the projection operator $\mathcal{E}_Q$, where
$$
\mathcal{E}_{Q} = \dfrac{1}{2} \cop_{\Theta}\aop_{\Theta} \left( 3 - \cop_{\Theta}\aop_{\Theta} \right)\ ,\quad \cop_\Theta=\sum_{j\geq2} c_j\cop_j\ ,
$$
with $\sum_{j\geq 2}|c_j|^2=1$.
Indeed, $\mathcal{E}_P$ and $\mathcal{E}_Q$ are commuting projectors and
$$
\bkmv{\Psi}{\mathcal{E}_P}{\Psi}=\bkmv{\Psi}{\mathcal{E}_Q}{\Psi}=\bkmv{\Psi}{\mathcal{E}_P\mathcal{E}_Q}{\Psi}=1\ .
$$
Therefore, $\ket{\Psi}$ is $(\mathcal{A}_1,\mathcal{A}_2)$-separable.
\item
Case when, in the first quantization approach, $\ket{\Psi}$ is entangled because the two parties cannot be simultaneously attributed two complete sets of properties.\\
In the second quantization approach, this case corresponds to $c_{11}\neq 0$ and $c_{1j}\neq0$ for at least one $j\geq 2$.
Indeed, the attribution of another complete set of properties beside the one associated with $\mathcal{E}_P$ in \eqref{comp2}, is equivalent to the existence of another projection
$$
\mathcal{E}_{Q} = \dfrac{1}{2} \cop_{\phi}\aop_{\phi} \left( 3 - \cop_{\phi}\aop_{\phi} \right)\ ,\quad \cop_\phi=\sum_{j\geq2} c^\phi_j\cop_j\ ,\quad \sum_{j\geq 2}|c^\phi_j|^2=1\ ,
$$
such that  $\bkmv{\Psi}{\mathcal{E}_Q}{\Psi}=1$, or $\mathcal{E}_Q\ket{\Psi}=\ket{\Psi}$, where $\hat{a}^\dag_\phi$ acting on the vacuum generates a single particle vector state $\ket{\phi}=\sum_{j\ge2}c^\phi_j\ket{\psi_j}$ orthogonal to $\ket{\psi_1}$
and $\hat{a}^\dag_\phi\hat{a}_\phi$ counts the number of Bosons in such a state.
Using the CCR \eqref{CCR},  one computes
$$
\hat{a}^\dag_\phi\hat{a}_\phi\ket{\Psi}=\Big(\sum_{j\geq2}c_{1j}\,(c^\phi_j)^*\Big)\,\ket{\phi}\ ,
$$
so that
$$
\mathcal{E}_Q\ket{\Psi}=\frac{1}{2}\Big(\sum_{j\geq2}c_{1j}\,(c^\phi_j)^*\Big)\,\ket{\phi}\neq \ket{\Psi}\ ,
$$
unless $\ket{\phi}=\ket{\Psi}$ which is possible only if $c_{11}=0$.
Thus, $\ket{\Psi}$ cannot be written in the factorized form \eqref{eq:thmben}; according to Proposition \ref{thm:main}, it is $(\mathcal{A}_1,\mathcal{A}_2)$-entangled.
\end{enumerate}

\subsection{From Second to First Quantization}

Proposition \ref{thm:main} states that in order to be separable with respect to a bipartition $(\mathcal{A}_1,\mathcal{A}_2)$,  vector states $\ket{\Psi}$ of two identical Bosons must be expressible in one of the following two forms
\begin{eqnarray}
\label{comp3}
\ket{\Psi}&=&\Big(\sum_{i\in I_1} c_{i}\cop_{i}\Big)\cdot \Big(\sum_{j\in I_2} d_{j}\cop_{j}\Big)\,\ket{0}\\
\label{comp4}
\ket{\Psi}&=&\Big(\sum_{i\in I_1} c^{(1)}_{i}(\cop_{i})^2\Big)\,\ket{0}\ ,\ \hbox{or}\quad \ket{\Psi}=\Big(\sum_{i\in I_2} c^{(2)}_{i}(\cop_{i})^2\Big)\,\ket{0}\ ,
\end{eqnarray}
where $\mathcal{A}_{1,2}$ are generated by $\{\hat{a}_i\,,\,\hat{a}^\dag_i\}_{i\in I_{1,2}}$ satisfying the CCR \eqref{CCR} with
$\hat{a}_i^\dag\ket{0}=\ket{i}$ forming an orthonormal basis in the single particle Hilbert space $\HI$.

In a first quantization setting, the vectors  $\hat{a}^\dag_i\hat{a}^\dag_j\ket{0}$, when $i\neq j$, are the
symmetric vectors
$$
\frac{\ket{i}\otimes\ket{j}+\ket{j}\otimes\ket{i}}{\sqrt{2}}\ .
$$
It thus follows, that
$\ket{\Psi}$ in \eqref{comp3} corresponds to the symmetrization of the tensor product of the orthogonal vectors
$$
\ket{\psi_{1}} = \sum_{i\in I_1}c_{i}\,\ket{i}\ ,\quad \ket{\psi_{2}} = \sum_{j\in I_2}\,d_{j}\,\ket{j}\ .
$$
Thus, any $(\mathcal{A}_1,\mathcal{A}_2)$-separable Bosonic state $\ket{\Psi}$ as in \eqref{comp3} is also separable in the first quantization approach for two orthogonal complete sets of properties can be attributed to  both its parties.

The $(\mathcal{A}_1,\mathcal{A}_2)$-separable two Bosons states \eqref{comp4} correspond to separable states in the first quantization approach when
only one coefficient $c^{(1)}_i$ or $c^{(2)}_i$ is not zero, in which case the two parties are in the same state and thus
possess the same complete set of properties.

As specified in Proposition \ref{thm:main}, all $(\mathcal{A}_1,\mathcal{A}_2)$-separable two-Fermion vector states $\ket{\Psi}$ that remain pure on the sub-algebra generated by the even components $\mathcal{A}_{1,2}^{ev}$, are also of the form \eqref{comp3},
where now the creation and annihilation operators $\hat{a}^\dag_j$ and $\hat{a}_j$ satisfy the CAR \eqref{CAR}.
Then, in the first quantization setting, the vectors  $\hat{a}^\dag_i\hat{a}^\dag_j\ket{0}$, when $i\neq j$, correspond to the anti-symmetric vectors
$$
\frac{\ket{i}\otimes\ket{j}-\ket{j}\otimes\ket{i}}{\sqrt{2}}\ .
$$
Therefore, $\ket{\Psi}$ in \eqref{comp3} corresponds to the anti-symmetrization of the tensor product of the orthogonal vectors
$$
\ket{\psi_{1}} = \sum_{i\in I_1}c_{i}\,\ket{i}\ ,\quad \ket{\psi_{2}} = \sum_{j\in I_2}\,d_{j}\,\ket{j}\ .
$$
Then, all two-Fermion $(\mathcal{A}_1,\mathcal{A}_2)$-separable vector states of the form \eqref{eq:thmben} are also separable in the first quantization approach.

\begin{remark}
\label{remfond2}
As already remarked, the first quantization (particle) approach to identical particle entanglement successfully applies to pure states only, whereas the second quantization (mode) approach covers the entire space of states of such systems. While from the previous discussion it may appear that the two approaches are equivalent for pure states, also in this case the mode-based approach is more general. Indeed, the mode description allows to address two-mode Bosonic states as
$$
\frac{\ket{0,1}+\ket{1,0}}{\sqrt{2}}=\frac{\hat{a}_1^\dag+\hat{a}^\dag_2}{\sqrt{2}}\ket{0}\ ,
$$
and to meaningfully claim that entanglement with the vacuum state is present, while the vacuum state is absent from the particle-based approach.
\end{remark}

\section{Conclusions}

In this paper we have considered an issue of entanglement theory which is still not settled,  namely, the characterization of non-local correlations in systems of identical particles, which is important not only in quantum many-body theory, but also in metrological
applications where one aim at using quantum entanglement in order to achieve sub shot-noise sensitivities \cite{benffmJPB,BeBr}.
In particular, we have compared two approaches to such an issue in the case of systems consisting of two Bosons or Fermions.
The first approach is based on the particle aspect of first quantization and on the attribution of complete sets of properties to both constituents, simultaneously.
Within this approach, which holds only for vector states, symmetrizing (anti-symmetrizing) tensor products of orthogonal single particle vector states of Bosons (Fermions), as required by their identity, does not lead to any entanglement.

The second approach is based on the mode aspect of second quantization and holds for all possible states of identical particle systems, be they pure or mixed.
Within this approach one can generalize the notions of local observables and of separable states by referring to pairs of sub-algebras constructed by means of
sets of creation and annihilation operators that commute or anti-commute when they are chosen from different sub-sets.
A difference immediately emerges between Bosons and Fermions for, in the latter case, there are anti-commuting observables while locality and separability are notions that make sense in relation to products of commuting observables.
These two approaches have been compared in the case of vector states and it has been proved that all vector states which result separable in the first quantization approach are such also in the second quantization setting.


\begin{thebibliography}{10}
\bibliographystyle{plain}

\bibitem{einapr1935}
A.~Einstein, B.~Podolsky, and N.~Rosen, Phys. Rev. {\bf 47}, 777 (1935).

\bibitem{horrhh2009}
R.~Horodecki, P.~Horodecki, M.~Horodecki, and K.Horodecki, Rev. Mod. Phys.
{\bf 81}, 865 (2009).

\bibitem{Zanardi1}
P.~Zanardi, Phys. Rev. Lett. {\bf 87}, 077901 (2001).

\bibitem{Cirac}
J.~Schliemann, J.I.~Cirac, M.~Ku\`s, M.~Lewenstein, and D.~Loss, Phys. Rev. A \textbf{64}, 022303 (2001).

\bibitem{Bruss}
K.~Eckert, J.~Schliemann, D.~Bruss, and M.~Lewenstein, Ann. Phys. \textbf{299}, 88 (2002).

\bibitem{ghigmw2002}
G.~Ghirardi, L.~Marinatto, and T.~Weber, J. Stat. Phys. {\bf 108}, 49 (2002).

\bibitem{GhiMar}
G.~C. Ghirardi, L.~Marinatto, Phys. Rev. A \textbf{70}, 012109 (2004).

\bibitem{Zanardi3}
P.~Zanardi, D.A.~Lidar and S.~Lloyd, Phys. Rev. Lett. {\bf 92}, 060402 (2004).

\bibitem{Narnhofer}
H.~Narnhofer, Phys. Lett. {\bf A310}, 423 (2004).

\bibitem{Viola1}
H.~Barnum, E.~Knill, G.~Ortiz, R.~Somma, and L.~Viola, Phys. Rev. Lett. {\bf 92}, 107902 (2004).

\bibitem{Viola2}
H.~Barnum, G.~Ortiz, R.~Somma and L.~Viola, Intl. J. Theor. Phys. {\bf 44}, 2127 (2005).

\bibitem{Amico}
L.~Amico, R.~Fazio, A.~Osterloh and V.~Vedral, Rev. Mod. Phys. {\bf 80}, 517 (2008).

\bibitem{benffm2010}
F.~Benatti, R.~Floreanini, and U.~Marzolino, Ann. Physics {\bf 325}, 924 (2010).


\bibitem{Abu}
M.C.~Tichy, F.~Mintert, and A.~Buchleitner, J. Phys. B 44, 192001 (2011).

\bibitem{Marmo}
J.~Grabowski, M.~Ku\`s and G.~Marmo, J. Phys. A 44, 175302 (2011).

\bibitem{benffm2012}
F.~Benatti, R.~Floreanini, and U.~Marzolino, Ann. Physics {\bf 327}, 1304 (2012).

\bibitem{bala}
A.P.~Balachandran, T.~R. Govindarajan, A.R.~de Queiroz and A.F.~Reyes-Lega, Phys. Rev. Lett. \textbf{110}, 080503 (2013).


\bibitem{BraRob}
O. Bratteli, D.W. Robinson, \emph{Operator Algebras and Quantum Statistical Mechanics 1,2}, Mathematical and Theoretical Physics, Springer 1997.

\bibitem{benffm2013}
F.~Benatti, R.~Floreanini and U.~Marzolino, Entanglement in fermionic systems and quantum metrology, preprint, 2013

\bibitem{benffmJPB}
F.~Benatti, R.~Floreanini and U.~Marzolino, J. Phys. B \textbf{44}, 091001 (2011).

\bibitem{BeBr}
F.~Benatti, D.~Braun, Phys. Rev. A \textbf{87}, 012340 (2013).





\end{thebibliography}
\end{document}